\documentclass[sigconf]{aamas} 

\usepackage{balance} 

\usepackage{graphicx} 
\usepackage{hyperref}
\usepackage{caption}
\usepackage[ruled,vlined,linesnumbered]{algorithm2e}
\usepackage{xcolor}
\usepackage{subcaption}
\usepackage{natbib}
\usepackage{multirow}
\usepackage{footnote}
\usepackage{amsmath}
\usepackage{amsthm}
\usepackage{booktabs}
\usepackage{xspace}
\usepackage{tabularx}

\renewcommand{\cite}{\citep}

\newcommand{\cppp}{\textsc {cppp}\xspace}

\newcommand{\mapf}{\textsc {MAPF}\xspace}
\newcommand{\kppmapf}{\textsc {kPPMAPF}\xspace}
\newcommand{\kpp}{\textsc {kPP}\xspace}
\newcommand{\eppmapf}{\textsc {ePPMAPF}\xspace}
\newcommand{\ekppmapf}{\textsc {ekPPMAPF}\xspace}
\newcommand{\fpp}{\textsc {fPP}\xspace}
\newcommand{\ppfpp}{\textsc {PPfPP}\xspace}
\newcommand{\agentgroupi}{\textit{AgGroup}$_i$\xspace}

\newcommand{\tuple}[1]{\ensuremath{\langle #1\rangle}}

\newcommand{\commentout}[1]{}

\SetKwProg{Fn}{Function}{:}{end}
\SetKwProg{Cls}{Class}{:}{end}
\SetKw{Params}{Parameters:}
\SetKwRepeat{DoWhile}{do}{while}

\commentout{
\newtheorem{myclaim}{myclaim}

\newtheorem{theorem}{Theorem}

\newtheoremstyle{exampstyle}
  {-.25pt} 
  {-.25pt} 
  {\itshape} 
  {} 
  {\bfseries} 
  {.} 
  {.5em} 
  {} 
\theoremstyle{exampstyle}
}
\newtheorem{theorem}{Theorem}
\newtheorem{definition}{Definition}
\newtheorem{corollary}{Corollary}

\newcommand{\plan}[1]{}

\newcommand{\rotem}[1]{}
\newcommand{\roni}[1]{}
\newcommand{\guy}[1]{}

\setcopyright{ifaamas}
\acmConference[AAMAS '26]{Proc.\@ of the 25th International Conference
on Autonomous Agents and Multiagent Systems (AAMAS 2026)}{May 25 -- 29, 2026}
{Paphos, Cyprus}{C.~Amato, L.~Dennis, V.~Mascardi, J.~Thangarajah (eds.)}
\copyrightyear{2026}
\acmYear{2026}
\acmDOI{}
\acmPrice{}
\acmISBN{}

\acmSubmissionID{823}

\title[Privacy Preserving MAPF]{Privacy Preserving Multi Agent Path Finding}

\author{Rotem Lev Lehman}
\affiliation{
  \institution{Ben Gurion University of the Negev}
  \city{Be'er Sheva}
  \country{Israel}}
\email{levlerot@post.bgu.ac.il}

\author{Roni Stern}
\affiliation{
  \institution{Ben Gurion University of the Negev}
  \city{Be'er Sheva}
  \country{Israel}}
\email{sternron@bgu.ac.il}

\author{Guy Shani}
\affiliation{
  \institution{Ben Gurion University of the Negev}
  \city{Be'er Sheva}
  \country{Israel}}
\email{shanigu@bgu.ac.il}

\begin{abstract}
    In the multi-agent path finding (\mapf) problem, a group of agents search in a graph for a path for each agent where no two paths collide.  
    While in all applications of \mapf the agents must not collide with each other, in some of them the agents may not wish to share their paths due to privacy constraints.   
    In this work, we formulate two types of privacy constraints for MAPF and propose algorithms that preserve them. The first type of privacy we consider is \emph{planning-level privacy}, which means that during planning, the agents cannot identify exactly the planned location of the other agents. 
    We propose a general framework for obtaining planning-level privacy, which works by adding \emph{mock agents} to the planning process. The second type of privacy we consider is \emph{execution-level privacy}, which is relevant when agents have limited sensing capabilities. Execution-level privacy is preserved if none of the agents is allowed to sense the location of the other agents during execution. We show how to adapt two popular \mapf algorithms, namely PIBT and LaCAM, such that they preserve execution-level privacy. 
    Lastly, we propose a post-processing technique that allows the agents to reduce the sum of costs of the returned solution without losing any privacy. We also implemented our algorithms and evaluated them empirically, showing that the proposed post-processing technique indeed improved cost significantly.     
\end{abstract}

\keywords{Multi Agent Path Finding, MAPF, Privacy}

\newcommand{\BibTeX}{\rm B\kern-.05em{\sc i\kern-.025em b}\kern-.08em\TeX}

\begin{document}


\pagestyle{fancy}
\fancyhead{}


\maketitle 


\section{Introduction}

The Multi-Agent Path Finding (MAPF) problem arises when multiple mobile agents must each find a path from their respective start locations to their goal locations on a shared graph. The primary requirement is that these paths are collision-free, meaning that no two agents occupy the same location or swap locations simultaneously. MAPF is motivated by a wide range of real-world applications, such as automated warehouse robotics, airport ground traffic management, and digital entertainment, where efficient and safe coordination of multiple agents is essential.

Most work on MAPF has assumed centralized control and complete information sharing between the agents. 
This is suitable for many use cases, such as warehouse robots controlled by a single operating company. 
In this work, we consider a different setting, where the agents are controlled by different entities that must collaborate to avoid collisions and optimize a shared objective function, yet still wish to keep some information private from each other. 
A primary motivation for such a setting is \emph{compartmentalization} due to security concerns. 
For example, consider an emergency response scenario in a large city, where multiple organizations—such as fire departments, police, and private security firms—deploy their own autonomous vehicles or drones to respond to incidents. While these organizations must coordinate to avoid traffic congestion and ensure rapid response, they may not wish to reveal the exact locations, routes, or priorities of their units due to operational security or privacy regulations.
More use-cases for privacy-preserving \mapf include coordinated trucks and drones of different logistics companies, coordinating routes of human taxi drivers, and coordinating the movement of robots in shared environments.
In this context, MAPF algorithms must enable safe and efficient multi-agent coordination while preserving the confidentiality of each organization's sensitive information. 

This form of privacy-preserving collaborative multi-agent planning has been studied in the context of other types of multi-agent planning problems, including Multi-Agent STRIPS~\cite{maliah2017collaborative,lev2022reducing,nissim2014distributed,brafman2015privacy}, Distributed Constraints Optimization (DCOP)~\cite{faltings2008privacy,greenstadt2007ssdpop} and Distributed Constraint Satisfaction Problem (DisCSP)~\cite{yokoo2002secure,729707}.
In this work we explore different types of privacy requirements one may consider in the context of MAPF and how one can achieve them. 

We distinguish between two types of privacy in MAPF: \emph{planning-level privacy} and \emph{execution-level privacy}. 
Planning-level privacy means that the agents cannot infer from the information they share during planning process any location and time any of the other agents are planning to visit. Execution-level privacy means that even if the agents are equipped with some sensing capabilities, they would still be unable to infer the location of the other agents during execution. 
Formally defining these types of privacy is the first contribution of this work.  

To obtain planning-level privacy, we propose a general framework that works by adding \emph{mock agents} to the planning process. Each real agent is associated with several mock agents, each has a fictitious start and goal location. Any MAPF algorithm can then be used to plan for all agents, including both real and mock agents. We discuss the type of strong privacy one obtains with this mechanism and its weaknesses. 
Then, we show how one can obtain execution-level privacy in two popular MAPF algorithms, namely PIBT~\cite{okumura2022priority} and LaCAM$^*$~\cite{okumura2024lacam3}, given prior knowledge about the sensing capabilities of the other agents.

To preserve privacy, the real agent cannot directly manipulate the planning process to optimize its own cost over the costs of the paths of the mock agents, as this may leak information about its true path. However, we show that the cost of the real agent can be improved in a post-processing step that does not affect the privacy guaranties. In this post-processing step, we identify \emph{safe zones} in the plan, which are locations and time steps where the real agent can deviate from its planned path without risking a collision or detection with other agents. We then use these safe zones to re-plan locally for the real agent, improving its cost while maintaining the overall privacy guaranties.


Finally, we implemented our privacy-preserving algorithms and evaluated them empirically, showing the relation between the amount of privacy we aim for and overall efficiency. Also, we show when our post-processing can be very effective and when it does not add significant gains. 
Overall, this work paves the way for future research on privacy-preserving MAPF, highlighting the trade-offs between privacy, efficiency, and solution quality in MAPF. 


\section{Background and Problem Setup}

\label{sec:background}


A classical MAPF problem with \(N\) agents is defined by a tuple \(\langle G, s, t \rangle\) where \(G = (V, E)\) is an undirected graph whose vertices are the possible locations agents may occupy, and every edge \((v, u) \in E\) represents that an agent can move from \(v\) to \(u\) without passing through any other vertex.
The functions $s$ and $t$ map each agent to its initial and desired destination locations, respectively. 
Time is discretized into time steps. In every time step, each agent can either \emph{wait} in its current location or \emph{move} to a location adjacent to its current location. 
A \emph{single-agent plan} for an agent \(i\) is a sequence of actions (wait or move) that if performed starting from \(s(i)\) will end up in  \(t(i)\). 
A \emph{joint plan} is a set of single-agent plans, one for each agent. 
For a joint plan \(\Pi\), we denote by \(\Pi_i\) its constituent single-agent plan for agent \(i\), and denote by $\Pi_i(t)$ as the vertex agent $i$ is planned to occupy at timestep $t$ according to $\Pi_i$.

A pair of agents \(i\) and \(j\) have a \emph{vertex conflict} in a joint plan \(\Pi\) if, according to their respective single-agent plans \(\Pi_i\) and \(\Pi_j\), both agents are planned to occupy the same vertex at the same time. Similarly, agents have a \emph{swapping conflict} in a joint plan if they are planned to swap locations over the same edge at the same time. A \emph{valid} solution to a MAPF problem is a joint plan without any conflict. 
Several solution cost functions have been proposed for MAPF. The two most common are sum of costs (SOC) and makespan, which are the sum and max, respectively, over the lengths of the single-agent plans in the solution. 
Finding cost-optimal solution for either cost function is NP-Hard~\cite{surynek2010anOptimization, yu2013structure}, and in directed graphs even finding a solution is NP-Hard~\cite{nebel2024computational}.

\plan{MAPF solvers (relevant ones)}
Nevertheless, many algorithms have been proposed for solving classical MAPF~\cite{felner2017search,surynek2022problem}.
A \mapf algorithm is \emph{complete} if it will eventually find a valid solution if one exists, and is \emph{optimal} if the found solution minimizes the cost function.
Some \mapf planners are complete and optimal, such as CBS~\cite{sharon2015conflict}, others only complete, such as LaCAM \cite{okumura2023lacam}, and others neither complete nor optimal yet are known to be very fast, such as PIBT~\cite{okumura2022priority}.
Anytime MAPF algorithms, such as LaCAM$^*$~\cite{okumura2024lacam3}, quickly find an initial solution and then continuously improve it as long as additional runtime is available. 
In this work, we build on PIBT and LaCAM$^*$ and therefore describe them briefly below.

\plan{Describe the PIBT algorithm}

PIBT~\cite{okumura2022priority} is a sequential \emph{configuration} generator, where a configuration here is an $N$-sized vector   representing the location of each agent.  
In each time step, PIBT accepts a configuration of the agents $Q^{from}$ and generates a new configuration $Q^{to}$ corresponding to a possible valid move of all agents. 
An initial priority is set at each timestep for all of the agents based on their distance to their goal, where the closest agent that hasn't already reached its goal gains the highest priority.
In order to determine $Q^{to}$, PIBT sequentially assigns a vertex to each agent while avoiding collisions. The order of assignment depends on the priority given to each agent, but before assigning a vertex $v$ to agent $a_i$, PIBT first checks if $\exists a_j \neq a_i: Q^{from}[a_j]=v$. If such $a_j$ exists, then it must apply that $Q^{to}[a_j] \neq v$ because of the collision avoidance. If $Q^{to}[a_j]$ was not assigned yet, PIBT is run on $a_j$ to try and move it from $v$, and $a_j$ gains the priority of $a_i$ to enable it to move lower priority agents in order to clear the way for $a_i$. If it was unable to move to another location, PIBT backtracks and tries to assign the next best vertex from the neighbors of $Q^{from}[a_i]$ for $a_i$.

\plan{Describe the LaCAM* algorithm}

LaCAM$^*$~\cite{okumura2024lacam3} is a two level-search algorithm, where each node in the high level search holds the configuration and a \emph{constraint tree}, and search throughout the nodes in a depth-first manner.
The low-level search gradually grows the constraint tree, and picks the next constraint node to explore, using a breadth-first manner.
Then, after picking both high-level and low-level nodes, LaCAM$^*$ generates a new configuration using a \emph{configuration generator}. The configuration generator must satisfy the constraints from the low-level node.
The authors of LaCAM$^*$ chose to use PIBT as a configuration generator, and so we also use it as such in our work.

\plan{Describe briefly the DisCSP setup}

In this work we also build on DisCSP~\cite{729707,yokoo2002secure}, and so describe it briefly together with privacy preserving solvers.

A DisCSP is a tuple $\langle A, X, D, C, \alpha \rangle$, where
\begin{itemize}
    \item $A = \{ a_1, a_2, \dots, a_m \}$ is a finite set of autonomous agents.
    \item $X = \{ x_1, x_2, \dots, x_n \}$ is a finite set of variables.
    \item $D = \{ D_1, D_2, \dots, D_n \}$ where each $D_i$ is a finite domain of possible values for variable $x_i$.
    \item $C = \{ c_1, c_2, \dots, c_k \}$ where each constraint $c_j$ is defined over a subset of variables and specifies the set of allowed tuples for those variables. Constraints can be defined over variables owned by the same agent or by different agents.
    \item $\alpha : X \rightarrow A$. The function $\alpha$ assigns each variable to exactly one agent, meaning (1) Agent $\alpha(x_i)$ is the owner of variable $x_i$. (2) Each agent controls the value assignment of its own variables.
\end{itemize}

\plan{Describe privacy preserving solvers for DisCSP}

The work in \cite{yokoo2002secure} proposes a secure privacy-preserving DisCSP algorithm that: (1) Uses public-key cryptography and multiple cooperating servers to perform the search process on encrypted information. (2) Ensures the distributed search (similar to chronological backtracking) does not leak private information of any agent. (3) Guarantees that neither other agents nor the servers can infer additional details about private variable values beyond the final agreed solution.
In other words, agents can collaboratively solve the constraint problem and reach agreement on a solution while keeping each agent’s private information completely confidential.

\paragraph{Problem setup and assumptions}
In this work, we consider a group of agents faced with a classical MAPF problem $\tuple{G,s,t}$. 
That is, each agent $i$ is initially located in an initial location $s(i)$ and aims to reach a target location $t(i)$ while avoiding collisions with the other agents. 
The environment is static and the agents' actions are deterministic. 
Each agent is fully aware of the underlying graph $G$, yet it does not know the initial and target location of the other agents. 
To coordinate, the agents may send messages to each other directly and immediately. 
Thus, without any privacy considerations the agents could solve the problem by sending their initial and target locations to one of the agents; have that agent reconstruct the MAPF problem and use any MAPF algorithm to solve it; share the solution with all agents; and finally have all agents safely execute it. 
The fundamental challenge we focus on in this work is that each agent does not wish the agents to know where it plans to visit in every timestep on its way to its target. 
That is, for a MAPF solution $\Pi$, agent $i$, and time step $t$, the \emph{private information} agent $i$ does not wish to disclose is the location $\Pi_i(t)$. 
Next, we consider two stages where such information can be revealed or inferred --- during planning and during execution --- and propose methods to preserve privacy in these stages.

\plan{Related work on Dec. MAPF solvers}
Our setting somewhat similar to prior work on distributed methods for solving \mapf~\cite{Velagapudi_2010,pianpak2019distributed,keskin2024decentralized,ho2020decentralized,wang2020walk}. 
One such approach~\cite{keskin2024decentralized,dergachev2021distributed} has each agent plan independently, coordinating and replanning online with other agents that enter its field of view. In that scheme, and in general in existing decentralized \mapf algorithms, agents may know the partial path of other agents if they are close to each other during the runtime. Thus, they do not provide any privacy guarantees, which is our primary objective.

\section{Planning-Level Privacy} 


During planning, agents may be able to infer some knowledge about the private information of other agents by analyzing the messages sent by each agent. 
We denote by $\mu$ the set of messages passed between the agents during the planning process. 
The \emph{agents' belief} about agent $i$ at time $t$, denote  $b_i(\mu,t)$, is the set of locations agent $i$ may occupy at time $t$ according to $\mu$. 
In other words, the agents' belief about agent $i$ captures the uncertainty of the other agents regarding the true location of $i$ at timestamp $t$. 
The \emph{belief state} of the agents is the vector $b=(b_1,\ldots, b_N)$ containing all the agents' beliefs. 

\begin{definition}[k-Privacy]
    \label{def:k-privacy}
    A belief state $b$ is said to preserve \emph{k-Privacy} if for every agent $i$ and time step $t$, it holds that $|b_i^{t}(\mu)| \geq k$.    
\end{definition}
Intuitively, this means that for every agent $i$ and time step $t$, 
the messages sent during planning are not sufficient to allow the other agents to narrow down the possible location of agent $i$ to fewer than $k$ possible locations. 
Next, we consider the problem of finding a solution to a given MAPF problem while preserving $k$-privacy during the planning process. We call this problem the k-Privacy Preserving \mapf (\kppmapf) problem. 
Note that \mapf is a special case of a \kppmapf problem, where $k=1$.

\subsection{k-Privacy Preserving \mapf Planner}

Now we describe the k-Privacy Preserving \mapf Planner (\kpp), which is a general algorithm for solving the \kppmapf problem. 
\kpp works by having each agent $i$ create a set of $k-1$ \emph{mock agents}, each associated with a unique pair of initial and target locations. 
Then, the agents collaboratively solve a larger MAPF problem that includes non-conflicting single-agent plans for both real and mock agents. Each (real) agent then follows the single-agent plan created for it, discarding the plans created for the mock agents. 
Next, we describe \kpp in more details.

In \kpp, the agents start by designating randomly one of the agents to be the \emph{planning agent}. 
Then, each agent $i$ performs the steps described in 
Algorithm~\ref{alg:kPP}. 
First, it creates a set of $k$ pairs of initial and target locations, denoted \agentgroupi. 
This set includes the pair $(s(i),t(i))$, i.e., the agent's real initial and target locations, as well as $k-1$ additional pairs. 
These pairs must be unique, i.e., different from $(s(i),t(i))$, from each other, and from the pairs used by the other (real) agents. 
Then, the agent shuffles these pairs randomly and shares them with all other agents (line~\ref{line:shuffle-broadcast} in Alg.~\ref{alg:kPP}). 
Next, agent $i$ waits until all agents have broadcast will receive the agents groups of all agents. 
At this point, if $i$ is the designated planning agent, then it creates a MAPF problem with the initial and target locations in all agent groups (a total of $k\cdot N$ ``agents''). Then, it calls any off-the-shelf MAPF solver to solve this MAPF problem, and broadcasts the solution to all the agents. 
If $i$ is not the planning agent, then it waits for the solution to be broadcasted by the planning agent. 
Finally, agent $i$ extracts from the solution created by the planning agent only the single-agent plan starting from $s(i)$ and ending in $t(i)$  (line~\ref{line:extract-solution}).



\begin{algorithm}[b!]
		\caption{\kpp for agent $i$}
            \label{alg:kPP}
        \agentgroupi $\gets (s(i), t(i))$ \\
        \For{$j=1$ to $k$}{
            $s(i_j), t(i_j)\gets$ choose unique initial and target locations \nllabel{line:choose-sg}\\
            Add $(s(i_j), t(i_j))$ to \agentgroupi \\
        }
        Shuffle \agentgroupi and broadcast it to all agents \nllabel{line:shuffle-broadcast}\\
        Wait for all agents to broadcast their agents group \nllabel{line:wait} \\
        \eIf{$i$ is the designated planning agent}{
            $\Pi\gets$ solve for all agent groups \nllabel{line:solve}\\
            Broadcast $\Pi$ to all agents \\
        }{
            $\Pi\gets $ Wait for a solution from the planning agent\nllabel{line:wait-solution} \\
        }
        Extract single-agent plan for $(s(i),t(i))$ \nllabel{line:extract-solution}\\
\end{algorithm}

\plan{Give an example that explains the concept of k privacy (not just k=1)}
Figure~\ref{fig:example-kpp}(left) illustrates a \kppmapf problem with two agents and the solution created for it by \kpp for $k=2$. 
The mock agents are marked in green and orange and the dashed lines mark the plans created for them.

\begin{figure*}[h]
    \centering
    \begin{subfigure}[t]{0.23\textwidth}
    \centering
      \includegraphics[width=1\linewidth]{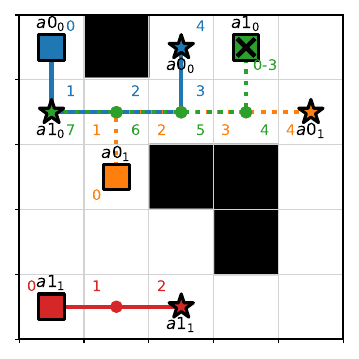}
      \caption{\kppmapf problem and the solution created for it by \kpp.}
      \label{fig:example-kpp}
      \Description{An example that shows the concept of k-privacy and the \kpp algorithm.}
    \end{subfigure}\hspace{1em}
    \begin{subfigure}[t]{0.23\textwidth}
    \centering
      \includegraphics[width=1\linewidth]{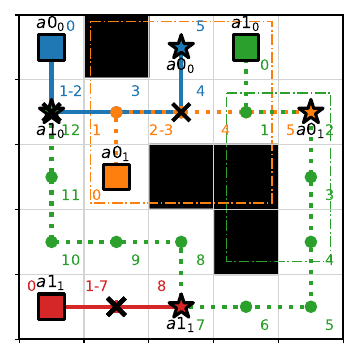}
      \caption{\ekppmapf problem. Dashed squares represent the FoV of $a_{0_1}, a_{1_0}$ at time step 3 of the plan.}
      \label{fig:example-fpp}
      \Description{An example that shows an \ekppmapf problem and the \fpp algorithm.}
    \end{subfigure}\hspace{1em}
    \begin{subfigure}[t]{0.23\textwidth}
    \centering
      \includegraphics[width=1\linewidth]{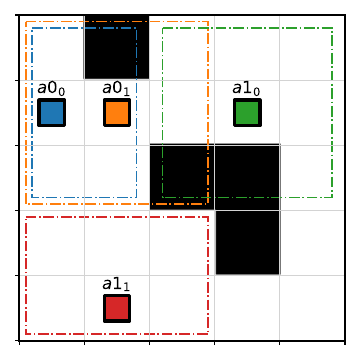}
      \caption{FoV of each agent at time step 1 of the plan from figure~\ref{fig:example-fpp}.}
      \label{fig:example-fpp_fov1}
      \Description{An example that shows the FoV of the agents at timestep 1.}
    \end{subfigure}\hspace{1em}
    \begin{subfigure}[t]{0.23\textwidth}
    \centering
      \includegraphics[width=1\linewidth]{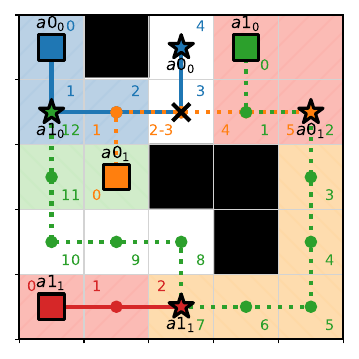}
      \caption{\ppfpp solution with initial and extended safe zones at time step 2 of the \fpp plan.}
      \label{fig:example-ppfpp}
      \Description{An example that shows \ppfpp safe zones.}
    \end{subfigure}\hspace{1em}
  \centering
  \caption[\kppmapf and \ekppmapf examples]{Examples of \kppmapf (\ref{fig:example-kpp}), \ekppmapf (\ref{fig:example-fpp}), FoV (\ref{fig:example-fpp_fov1}), and of \ppfpp (\ref{fig:example-ppfpp}) with 2 agents ($a_0$ and $a_1$) and $k=2$. Squares represents the initial location of each agent and stars represent the desired destination. The real agent of each group has a full line in the path and the mock agent has a dotted line. An X on the path of an agent marks that the agent waits in place in that spot for at least 1 timestep. The timestep of vertices in the plan appears near the vertex. In figure~\ref{fig:example-ppfpp} the blue and red vertices correspond to the initial safe zones for $a_0$ and $a_1$, respectively, while green and orange correspond to the extended safe zones.}
  \label{fig:examples}
  \Description{An example that shows the concept of k-privacy and the \kpp algorithm.}
\end{figure*}

\subsection{Choosing the Mock Agents}
\label{sec:choosing_mock_agents}
A critical step in \kpp is how each agent chooses the initial and target locations of its $k-1$ mock agents (line~\ref{line:choose-sg}).

\plan{Define a collision between assignments}

We consider an assignment of mock agents to have a collision between two agents $i, j$, if there exists a tuple (real or mock) $\langle s_{i_m}, g_{i_m} \rangle$ in the group of locations of agent $i$, and a tuple (real or mock) $\langle s_{j_n}, g_{j_n} \rangle$ in the group of agent $j$, where $s_{i_m} = s_{j_n} \lor g_{i_m} = g_{j_n}$.
Note that if any collision is found between two assignments, we cannot find a valid plan for the overall \kppmapf.
Also note that if two agents collide in their published groups of locations, and one of them ($i$) changes its previously published set $S$ by replacing a subset of locations $L \subset S$ to be another set of size $|L|$, then all other agents can understand that its real tuple $\langle s(i), g(i) \rangle \in S \setminus L$. This reduces $i$'s privacy to $k-|L|$ privacy instead of the wanted $k$ privacy.
So, in order to keep the needed $k$ privacy criteria, the algorithm used to pick mock agents need to publish the actual groups of locations only once, and to not have collisions.
We will now consider different options for selecting mock agents and discuss their pros and cons.

\subsubsection{Choosing Mock Agents Randomly}

\plan{analyze the probability of the probabilistic approach to fail.}
One approach to do so is to have each agent choose its group of locations randomly. 
A limitation for this approach is that two agents may collide in their assignments, as seen above.
We will now calculate the probability of such collision to occur.
The probability of no collision in the assignment of mock agents for a specific pair of agents in a graph with $|V|$ vertices, and a desire for $k$-privacy, when each agent chooses mock agents randomly is calculated in the following way:
Each agent $i$ has a predefined real tuple $\langle s(i), g(i) \rangle$, which do not collide with other agents real tuple. Then each agent chooses $k-1$ mock tuple, i.e. $2(k-1)$ vertices of the $k-1$ mock tuples from the remaining $|V|-1$ vertices, resulting in the following probability:
$\frac{\binom{|V|-1-2k}{2(k-1)}}{\binom{|V|-1}{2(k-1)}}$.
Therefore, The probability of no collisions in the assignment of $N$ agents is approximately:
\[\left( \frac{\binom{|V|-1-2k}{2(k-1)}}{\binom{|V|-1}{2(k-1)}} \right)^{\binom{N}{2}} \xrightarrow[\;]{N \to \infty, k > 1}\,0\]

As can be seen, this approach can lead to collisions with higher probability as the number of agents planning grows.
Next, we explain some other methods for choosing mock agents.

\subsubsection{Choosing Mock Agents using a DisCSP}

\plan{Model the problem of selecting mock agents as a DisCSP problem}

We suggest another approach for choosing the mock agents in a way that will both not collide with other agents and will preserve the privacy of each agent by not revealing the real (start, goal) tuple of each agent.
We can model the problem of choosing the mock agents in a privacy-preserving manner as a DisCSP.
Each agent will have $2(k-1)$ variables, representing the start and goal of its mock agents, and two additional variables for the real start and goal.
The domain of each variable is all possible locations in the graph ($V$).
The public constraints are for agents to avoid conflicts.
The private constraints of each agent are to ensure that its real start and goal are assigned to it.
Solving this DisCSP will ensure that every agent will have its real start and goal in its k-sized group, and there will be no conflicts, while making sure that this assignment does not reveal any private information.

The wanted privacy model here is to hide the private constraints of the agents, and make sure other agents cannot distinguish between the real tuple of each agent and its mock tuples.
This can be solved securely using the method in~\cite{yokoo2002secure}, since they make the search process in remote servers that do not know the values of the variables since they are encrypted, and the agents do not know any knowledge about the failed assignments that happens during the search process (which could reveal the collision between two agents groups), since they are not part of the search and only receive the final assignments for their groups.
Hence, we can conclude that no agent can distinguish the other agents real start and goal vertices from their set of disclosed locations.

Solving this DisCSP problem can be very expensive in large graphs, with many agents and large values of k, so one can think of using a different approach that we explain now.

\subsubsection{Mock Agents Dispatcher}

\plan{Explain the dispatcher model}

We suggest yet another approach of choosing the mock agents, by using an \emph{external dispatcher} agent whose sole role is to centralize the process of generating a unique initial and target location pairs when requested to do so by an agent (line~\ref{line:choose-sg}). 
This external dispatcher agent knows the real initial and target locations of all agents and makes sure they do not collide.

\plan{Say that using dispatcher is not too bad since it only knows the start and goals and not the path. The privacy of an agent will be entirely compromised only if the dispatcher "works together" with some agent.}

This can be viewed as some loss of privacy.
However, it is external to the path planning agents, and so it does not know the plans the agents end up selecting.
Thus, unless the dispatcher collaborates with one of the agents, the amount of privacy lost is limited to the dispatcher knowing the start and goal locations of the agents.

\plan{Motivate the dispatcher using the long-list-of-tasks use case}

In order to motivate the use of the external dispatcher, consider a different problem setting, where the dispatcher maintains a large pool of tasks and assigns $k$ candidate tasks to each agent, without knowing which of the $k$ tasks the agent will actually commit to.
This provides a natural way to generate indistinguishable alternatives without requiring plan-level information.

The \kpp algorithm's behavior is independent of the mock agents selection method in use, and one can use whichever method it prefers.
We chose to implement our experiments in section~\ref{sec:experimental-results} using the external dispatcher method because of its simplicity.

\commentout{

Sophisticated consensus mechanisms may be applied to result this problem probabilistically. In our work, we propose to use an external dispatcher agent whose sole role is to centralize the process of generate a unique initial and target locations pair when requested to do so by an agent (line~\ref{line:choose-sg}). 
This external dispatcher agent knows the real initial and target locations of all agents, and can be viewed as some loss of privacy. Yet it is external to the path planning agents, and can be implemented in a secure way.
}

\subsection{Theoretical Analysis of \kpp}
\label{sec:kpp-theory}

Due to space constraints, we prove informally, the main theoretical property of \kpp, namely, that it outputs valid solutions that preserves $k$ privacy. 

\begin{theorem}[\kpp Solves \kppmapf problems]
\label{theorem:kpp-solves-kppmapf}
    If the designated planning agent in \kpp returns a valid \mapf solution, then \kpp returns a valid \kppmapf solution.
\end{theorem}
\textbf{Proof outline:}
A valid \mapf solution ensures there are no conflicts between all agents. This includes conflicts between the real agents and also conflicts between a real agent and the other (mock) agents in its agent group. 
This means that at any point in time the agents in an agent group occupy $k$ different locations. All messages sent between the agents are agnostic to who is the main agent and who is the mock agents within an agent group. Thus, the belief state will be the same regardless of who is the real agent. Thus, a conflict-free solution has been found and $k$-privacy preserved.

Unfortunately, \kpp is not complete, since there could be a randomized agent group initial locations and targets that would make the original \mapf problem solvable but the new problem unsolvable. Similarly, \kpp is not optimal even if the designated planning agent uses an optimal \mapf algorithm.

\section{Execution-Level Privacy}
The agents in \mapf may be equipped with sensors that allow them to sense the location of nearby agents. 
In such cases, the agents risk revealing their locations during execution, 
even if they execute a \mapf solution generated by a planning-level privacy preserving \mapf algorithm. Next, we formally describe the problem that arises in these cases and discuss how to ensure privacy is still preserved, even during execution.

\subsection{Conflicting Field of Views}
\label{sec:eppmapf-def}

We formalize the sensing capability of an agent $i$ by a 
\emph{Field of View} (FoV) function $F_i$ that maps every possibly location $v$ of agent $i$ to the locations agent $i$ can sense when situated in $v$. 
If an agent can sense the location of another agent during execution, i.e., it enters if FoV, then that agent's privacy has been compromised. 
We refer to this as as FoV conflict and define it formally as follows. 
\begin{definition}[FoV Conflict]
    A pair of single agent plans $\Pi_i$ and $\Pi_j$ assigned to agents $i$ and $j$, respectively, are said to have a FoV conflict if there exists $t$ such that either $\Pi_i(t)\in FoV_j(\Pi_j(t))$ or 
    $\Pi_j(t)\in FoV_i(\Pi_i(t))$.    
\end{definition}

An \eppmapf problem is defined by a tuple \(\langle P^*, F \rangle\) where $P^*$ is a classical \mapf problem, and $\{F_1,\ldots,F_N\}$ is the agents' FoV functions. 
    A solution $\Pi$ for an \eppmapf problem is called valid if it is a valid solution for the underlying \mapf problem ($P^*$) and it has no FoV conflict.
It is hard to motivate preserving privacy during execution but not during planning. 
Thus, for the rest of this work we require a valid \eppmapf solution to also preserve $k$-privacy.

\commentout{
\plan{Add an example of a problem with field of view and a solution for it.}
Figure~\ref{fig:example-fpp} illustrates an example of a \ekppmapf problem and a valid solution for it for $N=2, k=2, \forall v:F_0(v)=F_1(v)=\{ u | |u.x - v.x| \leq 1 \And |u.y - v.y| \leq 1 \}$.
As can be seen, The plan from figure~\ref{fig:example-kpp} is not valid for the \ekppmapf problem, since the paths in that plan has $FoV$ conflict. 
As a result a different path for $a_{1_0}$ is chosen, but in timestep 2, agent $a_{0_1}$ must wait-in-place, since otherwise it will collide with agent $a_{1_0}$. As can be seen, the $FoV$ of the two agents do not contain the other agent, and so the plan is now valid. Notice that, for example in timestep 1, the $FoV$ of agents $a_{0_0}$ and $a_{0_1}$ contain the other agent, but it is fine since they are both in the same agent-group.

\subsection{Separating the Field of Views} 

While the \eppmapf problem is defined for every privacy preserving \mapf problem, in this paper we explore it for the \kppmapf problem.
We define the new problem of Execution k Privacy Preserving \mapf (\ekppmapf).
An \ekppmapf problem is defined by a tuple \(\langle G, s, t, k, F \rangle\) where:

\begin{itemize}
  \item \(G, s \text{ and } t\) are the same as the classical \mapf definition in section~\ref{sec:background}.
  \item \(k\) is the amount of privacy preserved in the plan, as defined in section~\ref{sec:kppmapf-def}.
  \item \(F_i\) is a FoV function for each agent, as defined in section~\ref{sec:eppmapf-def}.
\end{itemize}
}

\begin{definition}[Runtime k-Privacy]
    \label{def:runtime-k-privacy}
    An \eppmapf solution $\Pi$ preserves \emph{Runtime k-Privacy} if 
    the belief state generated from finding $\Pi$ preserves k-Privacy and there are no $FoV$ conflicts in $\Pi$.  
\end{definition}
The \ekppmapf problem is defined as the problem of finding a runtime $k$-privacy solution for a given MAPF problem, desired planning-level privacy $k$, and FoV functions. 

Definition~\ref{def:runtime-k-privacy} is important, since it makes the runtime stage, where agents are aware of the joint plan for all agents, k-privacy preserving, since if the agent $a_i$ saw $a_j$ at vertex $v$ during the runtime, at time $t$, it ruins the entire belief state of $a_j$ not only for $b_j^t$, but also for all $b_j^{t'} | t' \neq t$, since $a_i$ can now understand in hindsight that the path related to the $v$ that $a_j$ is traversing in the joint plan, is its real path.
In a different case, if agent $a_i$ traverses near vertex $v$ at time $t$ where $v \in b_j^t(\Pi)$ but $a_i$ can see that $a_j$ is not in $v$ at time $t$, then $a_i$ can reduce the belief state of $a_j$ to be ${b_j^t}' \gets b_j^t \setminus \{v\}$, and, again, it can do so for the entire path that is related to $v$ in the joint path, and by doing so, it reduces the privacy of $a_j$ from $k-Privacy$ to $(k-1)-Privacy$.
This is why it is important to preserve $Runtime \ k-Privacy$ when there are sensors in the agents and $k-Privacy$ is needed.

Figure~\ref{fig:example-fpp} illustrates an example of a \ekppmapf problem and a valid solution for it for $N=2, k=2, \forall v:F_0(v)=F_1(v)=\{ u | |u.x - v.x| \leq 1 \And |u.y - v.y| \leq 1 \}$.
As can be seen, The plan from figure~\ref{fig:example-kpp} is not valid for the \ekppmapf problem, since the paths in that plan has $FoV$ conflict. 
As a result a different path for $a_{1_0}$ is chosen, but in timestep 2, agent $a_{0_1}$ must wait-in-place, since otherwise it will collide with agent $a_{1_0}$. As can be seen, the $FoV$ of the two agents do not contain the other agent, and so the plan is now valid. Figure~\ref{fig:example-fpp_fov1} shows the $FoV$ of the agents at time step 1. Notice that the $FoV$ of agents $a_{0_0}$ and $a_{0_1}$ contain the other agent, but it is fine since they are both in the same agent-group.

Notice that \kppmapf problem is a sub-problem of an \ekppmapf problem, where \(\forall a_i \forall v \ F_i(v)=\{v\}\).
\commentout{
Obvious
\begin{corollary}
\label{cor:kppmapf-sub-ekppmapf}
    \kppmapf problem is a sub-problem of an \ekppmapf problem, where \(\forall a_i \forall v \ F_i(v)=\{v\}\)
\end{corollary}
}

A straightforward \ekppmapf planner can be obtained by extending the \kpp, resulting in the Field-of-View $k$-Privacy Preserving Planner (\fpp).
\fpp follows the standard \kpp framework but modifies the underlying \mapf planner to incorporate additional $FoV$ conflicts that prevent agents from entering each other’s fields of view.
These $FoV$ conflicts are enforced only between different agent groups, and not between sub-agents within the same group, since sub-agents represent hypothetical alternatives of a single real agent.
Consequently, visibility among sub-agents of the same group does not compromise privacy during execution.

Note that when selecting the agent groups in line~\ref{line:choose-sg} of algorithm~\ref{alg:kPP}, as seen in section~\ref{sec:choosing_mock_agents}, we need to consider a broader collision definition, which also takes into account that in each tuple of start, goal vertices the vertices are not in the $FoV$ of other agents tuples.

            


\subsection{Theoretical Analysis of \fpp}
\label{sec:fpp-theory}

By construction, \fpp preserves runtime $k$ privacy. 
As discussed in section~\ref{sec:kpp-theory}, \kpp is neither complete not optimal, and so \fpp is also neither complete nor optimal. 
Section~\ref{sec:enhancing} proposes a post-processing plan improvement step that can be performed on top of \fpp to reduce the cost of the solution returned without compromising the required runtime $k$ privacy requirement.


\subsection{Implementing \fpp} 
\label{sec:subsolvers}

\fpp requires the $Sol$ used by it to avoid FoV  conflicts with agents of different agent groups. In order to support it, one can alter available \mapf planners to support the $FoV \ conflict$ avoidance.
Of course, there will be no conflict if $a_{i_k}$ and $a_{i_q}$ which are in the same agent group $a_i$ will be in each other's FoV, and so the sub-solver we use should avoid considering this as a conflict. 
Next, we describe how to do so in two state-of-the-art \mapf planners: $PIBT$ \cite{okumura2022priority} and $LaCAM*$ \cite{okumura2023lacam,okumura2023lacam2,okumura2024lacam3}.


\plan{Implementation in PIBT}
PIBT is built to have a priority inheritance stage and a backtracking stage, inside a loop. The priority inheritance means that if an agent $a_i$ with higher priority wants to go to a vertex $v$ where an agent with lower priority $a_j$ resides currently, then it will start a PIBT session for $a_j$ in order for it to move from $v$ to some other vertex and let $a_i$ go to $v$, and in the process move other agents, even with higher priority than $a_j$, but with lower priority than $a_i$. 
If it cannot move to another location then it will backtrack and $a_i$ will have to find a different location to go to.
In order to support $FoV$ conflicts resolution, we don't move only the agent $a_j$ from its location, but all agents in the set $\{ a_k | \Pi_k(t) \in F_i(v) \lor v \in F_k(\Pi_k(t)) \}$ must move to clear the way for $a_i$, since otherwise there will be a $FoV$ conflict at time t. To do so we run PIBT on all of the set of agents in the field of view of $v$, and if one of them fail (cannot move from its current location), we backtrack all of the set (since they do not longer have the priority of $a_i$).

LaCAM$^*$ builds on PIBT to generate configuration. 
Thus, adapting it to support execution-level privacy is straightforward - simply use the adapted PIBT described above. 
We note that LaCAM$^*$ also includes a slight change to PIBT to support a \emph{swap} operation which increases performance. 
We did not implement this in our work as it requires additional modifications that we leave to future work. 

\commentout{
\subsection{Supporting \fpp in CBS}
\plan{Note - this subsection should be written only if we complete the CBS code before we need to send the article}
\plan{Implementation in CBS}
}

\section{Improving Solution Quality}
\label{sec:enhancing}

A key property of a runtime $k$-privacy solution is that for every timestep $t$ and agent $i$, all the locations planned for agent $i$ and associated $k-1$ mock agents will not be in the FoV of any other agent. 
We refer to such vertices, i.e., the vertices that are guaranteed to not be sensed by any other agents, as the \emph{safe zone} of agent $i$ at timestep $t$. 
Each agent can choose a different plan for itself without coordinating with the other agents as long is the new plan only occupies vertices in corresponding safe zones. 
This can be beneficial, as the agent may now prioritize its single-agent plan over those of its $k-1$ mock agents, potentially obtaining lower SOC for itself, without compromising privacy. 
Next, we discuss how this can be done. 

\subsection{Safe zones in \ekppmapf plans}

In this section we define notions to be used in the post process for improving the plans originated from running \fpp.

\plan{Define the initial safe-zones of an agent group}

\begin{definition}[Agent Group's FoV]
    \label{def:agent-group-fov}
    An agent group $a_i$'s FoV in timestep $t$ on a given plan $\Pi$ generated by \fpp on an \ekppmapf problem P, denoted $FoV_i^t(\Pi)$, is the set $\{v | \exists a_{i_j} v \in F_i(\Pi_j(t))\}$.
\end{definition}

In other words, $FoV_i^t(\Pi)$ is the set of vertices that are in the FoV of any agent in the group of $Mock_i$ in time $t$ using the plan $\Pi$.

\begin{definition}[Initial Safe-Zone]
    \label{def:initial-safe-zone}
    The initial safe-zone of an agent group $a_i$ in a specific timestep $t$ on a given plan $\Pi$ that was generated from running \fpp on an \ekppmapf problem P, denoted $IS_i^t(\Pi)$, is the set $\{v | F_i(v) \subseteq FoV_i^t(\Pi)\}$.
\end{definition}

In other words, $IS_i^t(\Pi)$ is the set of vertices that all of the vertices in their FoV are in the agent group's FoV.

\begin{definition}[Separated Safe Zones]
    \label{def:safe-zone-safeness}
    A set of safe-zones $S(\Pi) = \left\{ S_i(\Pi) | \forall i \in A \right\}$, is called separated, if and only if: \\$\forall{t}\forall{i \in A} \forall{v \in S_i^t(\Pi)} \forall{j \neq i} \nexists{u \in S_j^t(\Pi)}: u \in F_i(v) || v \in F_j(u)$.
\end{definition}

In other words, A set of safe-zones $S(\Pi)$ is separated if all of the vertices in each $S_i(\Pi)$ are out of the field of view of each other in each timestep.

\begin{definition}[Symmetric Field of View function]
    A Field of View function $F$ is considered symmetric, if and only if $\forall v \in V \forall u \in F(v): v \in F(u)$.
\end{definition}

\begin{definition}[PathCost]
    A single agent $i$ path cost, denoted $PathCost(\Pi_i)$, is the amount of actions performed by agent $i$ in $\Pi_i$ until it reaches $g(i)$ and stays there.
\end{definition}

\begin{definition}[Real Agent SoC]
    \label{def:real-agent-soc}
    A real-agent SoC of an \ekppmapf plan $\Pi$, denoted $RSoC(\Pi)$, is calculated as follows: $RSoC(\Pi) \gets \sum_{a_i \in A}PathCost(\Pi[a_i.real\_agent])$. 
\end{definition}

In other words, $RSoC(\Pi)$ is the sum of costs of the real agents' paths.

\plan{Show in the example of the start on this section which vertices are the safe-zones and which are the extended safe-zones}

\subsection{Post-Processing \fpp algorithm}

\plan{write a pseudo-code for it}
\begin{algorithm}[b!]
		\caption{Post-Processing (\fpp) Solutions (\ppfpp)}
            \label{alg:ppfpp}
            \SetKw{KwInput}{Input:}
            \SetKw{KwOutput}{Output:}
            
            \KwInput{$\Pi$: An \ekppmapf solution}\\
            \KwOutput{$\Pi'$ refined for all agent groups}
            
		\Fn{PPfPP}{
                $IS(\Pi) \gets \text{initial safe zone as defined in definition~\ref{def:initial-safe-zone}}$ \;
                $ES(\Pi)\gets ExtendSafeZone(IS(\Pi))$ \;
                \label{line:ppfpp-es}
                \ForEach{$a_i \in A$}{
                    $a_i.FindPathForRealAgent(ES_i(\Pi))$ \;
                    \label{line:FindPathForRealAgent}
                }
            }

            \Fn{ExtendSafeZone(IS)}{
                \label{line:ppfpp-extend-safe-zone}
                ES $\gets$ IS  \;
                \label{line:ppfpp-extend-safe-zone-initialize}

                \ForEach{$t \in T$}{
                    \DoWhile{$\neg done$}{
                        $done \gets TRUE$\;
                        \ForEach{$a_i \in A$}{
                            \label{line:alg_ppfpp_extend_safe_zone__foreach_agent}
                            $picked \gets a_i.ExtendSafeZone(ES^t)$  \;
                            \label{line:ppfpp-a_i-extend-safe-zone}
                            $done \gets done \land \neg picked$  \;
                        }
                    }
                }
                \Return ES \;
            }

\end{algorithm}

\plan{Explain the pseudo-code}

Algorithm~\ref{alg:ppfpp} is run by a shared planner. 
It first calculates the $IS$, as defined in definition~\ref{def:initial-safe-zone}.
Then using it, it calculates the $ES$, which is the extended safe zone. 
It does so by using $a_i.ExtendSafeZone(ES^t)$.
$a_i.ExtendSafeZone(ES^t)$ is a function for each agent, that defines the policy of the agent on picking the next vertex to add to its extended safe zone $ES_i^t$.
Then, each agent group $a_i$ can use $ES_i$ to calculate the best single agent plan for its real agent that passes only inside of the $ES_i$ (this is the \textit{FindPathForRealAgent} method in line~\ref{line:FindPathForRealAgent} in Alg.~\ref{alg:ppfpp}). In our implementation, it is done by using $SIPP$~\cite{phillips2011sipp}, since it handles well planning with safe intervals, which in our case where the safe intervals from $ES_i$. 
\begin{definition}[Rules for Extending Safe Zone]
    \label{def:rules-for-extended-safe-zone}
    The function $a_i.ExtendSafeZone(ES^t)$ is a function that $a_i$ uses for extending the safe zone $ES_i^t$. It must pick one vertex $v$ to add to $ES_i^t$. It returns $TRUE$ if found $v$ that was added to $ES_i^t$, $FALSE$ otherwise, and it must follow the following rules:
    \begin{enumerate}
        \item $v$ is a neighbor of some $u \in ES_i^t$.
        \item $\forall a_j, v \notin ES_j^t$.
        \label{item:rules-for-extended-safe-zone--es-unique}
        \item $\forall a_j \neq a_i, \nexists u \in ES_j^t: u \in F_i(v) \lor v \in F_j(u)$.
        \label{item:rules-for-extended-safe-zone--out-of-fov}
        \item Using it must preserve the runtime k-Privacy of $\mu,\Pi$ after running it.
    \end{enumerate}
\end{definition}

For example, a function that qualifies for all of the rules above is the function $ExtendSafeZoneRandomly$ defined in algorithm~\ref{alg:extend-safe-zone-random}. Algorithm~\ref{alg:extend-safe-zone-random} can be implemented distributively, by setting the same random seed and the order of agents to pass by in the for loop at line~\ref{line:alg_ppfpp_extend_safe_zone__foreach_agent} to each agent for ensuring that they all result in the same $ES$.

\begin{algorithm}[b!]
    \caption{Extend safe zone randomly.}
    \label{alg:extend-safe-zone-random}
    
    \Fn{$a_i.ExtendSafeZoneRandomly(ES^t)$}{
        $X \gets \{ v | \exists u \in ES_i^t: (u,v) \in E \lor (v, u) \in E \} \setminus ES_i^t$ \;
        \ForEach{$a_j \neq a_i$}{
            $X \gets X \setminus ES_j^t \setminus ES_j^{t-1}$ \;
            $X \gets X \setminus \{ v | \exists u \in ES_j^t: u \in F_i(v) \lor v \in F_j(u) \}$ \;
        }
        \If{$|X| > 0$}{
            $v \gets RandomPick(X)$ \;
            $ES_i^t \gets ES_i^t \cup \{ v \}$ \;
            \Return $TRUE$ \;
        }
        \Return $FALSE$ \;
    }
\end{algorithm}

\plan{Show in the example added in beginning of the section how the path will change after running the post-process algorithm in two cases}

Figure~\ref{fig:example-ppfpp} shows an example for the initial and extended safe zones, and for the plan generated from the extended safe zones using the \ppfpp algorithm.
As can be seen, the \fpp plan in figure~\ref{fig:example-fpp} is not the same as the \ppfpp plan in figure~\ref{fig:example-ppfpp}.

\plan{First case - using the initial safe-zones of the agent group}

For example, $a_{0_0}$ waited at time step 1, in the \fpp plan, since otherwise it would have a conflict with $a_{0_1}$ which must wait at time step 2. Although, in the \ppfpp we can see that the initial safe zone of $a_0$ contains the steps needed for $a_{0_0}$ to continue without waiting until it reaches its goal. This reduces the path cost of $a_0$, which its real agent is $a_{0_0}$, to 4 instead of 5.

\plan{Second case - using the extended safe-zones of the agent group}

$a_{1_1}$ on the other hand, cannot continue to its goal on time step 2 based on its initial safe zone, since it contains only the vertex it steps on in the \fpp plan. But, the extended safe zone does contain its goal at time step 2, and so it can continue to it. This reduces the plan cost of $a_1$, which its real agent is $a_{1_1}$, to 2 instead of 8.

\subsection{Theoretical Analysis of \ppfpp}

In this section we will show that the extension of safe zones is safe, that \ppfpp returns valid \ekppmapf plans, and that the resulting plan cost is less or equal to the original plan.

\plan{Theorem for safeness}

\begin{theorem}[Initial safe-Zone Safeness]
    \label{theorem:initial-safe-zone-safeness}
    An initial safe-zone $IS(\Pi)$ calculated from definition~\ref{def:initial-safe-zone}, when $\forall a_i, a_j \in A \forall v \in V: F_i(v)=F_j(v) \And F_i \text{ is symmetric}$, is safe.
\end{theorem}
\commentout{
\begin{proof}
    Given an \ekppmapf problem P, a valid plan $\Pi$ that solves P, and was returned by using \fpp on P, an initial safe-zone $IS(\Pi)$ calculated from definition~\ref{def:initial-safe-zone}, and $\forall a_i, a_j \in A \forall v \in V: F_i(v)=F_j(v) \And F_i \text{ is symmetric}$.
    Let $a_i, a_j \in A| a_i \neq a_j, t \in T$.
    Also, let $v \in IS_i^t(\Pi)$.
    Assume negatively that $\exists u \in IS_j^t: \left( u \in F_i(v) \| v \in F_j(u) \right)$, without loss of generality, assume $u \in F_i(v)$.
    From definition~\ref{def:initial-safe-zone}, $u \in FoV_i^t(\Pi)$, hence from definition~\ref{def:agent-group-fov}, $\exists a_{i_p}: u \in F_i(\Pi_p(t))$, because that $F_i$ is symmetric, $\Rightarrow \Pi_p(t) \in F_i(u)$ and because $F_i = F_j \Rightarrow \Pi_p(t) \in F_j(u)$. Because of $u \in IS_j^t(\Pi)$, then $F_j(u) \subseteq FoV_j^t(\Pi) \Rightarrow \Pi_p(t) \in FoV_j^t(\Pi)$, and so from definition~\ref{def:agent-group-fov}, $\exists a_{j_q}: \Pi_p(t) \in F_j(\Pi_q(t))$. But $a_{i_p}$ and $a_{j_q}$ are of different agent groups, and so this is a FoV conflict $\Rightarrow$ from definition~\ref{def:runtime-k-privacy}, $\Pi$ does $not$ preserve Runtime k-Privacy, and so from definition~\ref{def:valid-ekppmapf-plan} $\Pi$ is not a valid \ekppmapf plan, in contradiction to the assumption that $\Pi$ is a valid plan that solves P.
    $\Rightarrow \nexists u \in IS_j^t: \left( u \in F_i(v) \| v \in F_j(u) \right) \Rightarrow$ from definition~\ref{def:safe-zone-safeness} $IS(\Pi)$ is considered safe.
\end{proof}
}

\begin{theorem}[Extended Safe-Zone Safeness]
    \label{theorem:extended-safe-zone-safeness}
    An extended safe-zone $ES(\Pi)$ calculated from running $ExtendSafeZone$ in line~\ref{line:ppfpp-extend-safe-zone} in algorithm~\ref{alg:ppfpp} on $IS$ calculated from definition~\ref{def:initial-safe-zone}, when the function $a_i.ExtendSafeZone(ES^t)$ from line~\ref{line:ppfpp-a_i-extend-safe-zone} follows the rules from definition~\ref{def:rules-for-extended-safe-zone}, when $\forall a_i, a_j \in A \forall v \in V: F_i(v)=F_j(v) \And F_i \text{ is symmetric}$, is safe.
\end{theorem}

\begin{theorem}[\ppfpp Solves \ekppmapf Problems]
    \label{theorem:ppfpp-solves-ekppmapf}
    If $ES$ was calculated using a $a_i.ExtendSafeZone(ES^t)$ function that follows the rules from definition~\ref{def:rules-for-extended-safe-zone}, and $\forall a_i, a_j \in A \forall v \in V: F_i(v)=F_j(v) \And F_i \text{ is symmetric} \And \forall u \in v.neighbors: u \in F_i(v)$, then \ppfpp returns valid \ekppmapf plans.
\end{theorem}

The proofs for theorems~\ref{theorem:initial-safe-zone-safeness}, theorem~\ref{theorem:extended-safe-zone-safeness}, and theorem~\ref{theorem:ppfpp-solves-ekppmapf} are pretty straightforward, and hence are attached in the supplementary material of this paper.

\plan{Explain that \ppfpp cannot be also run on \kppmapf plans}

Even though \kppmapf is a sub-problem of \ekppmapf, we cannot run \kppmapf plans inside the \ppfpp algorithm to improve their plan cost, since in Theorem~\ref{theorem:ppfpp-solves-ekppmapf} we must have the neighbors of each vertex inside the FoV function in order for it to return valid plans (to avoid $swapping$ conflicts).

\plan{Theorem for improving cost}

\begin{theorem}[\ppfpp Improves $RSoC(\Pi)$]
    \label{theorem:ppfpp_improves_rsoc}
    If $ES$ was calculated using a $a_i.ExtendSafeZone(ES^t)$ function that follows the rules from definition~\ref{def:rules-for-extended-safe-zone}, and $\forall a_i, a_j \in A \forall v \in V: F_i(v)=F_j(v) \And F_i \text{ is symmetric}$, and $\Pi'$ is the result of running \ppfpp on a valid \ekppmapf plan $\Pi$, then $RSoC(\Pi') \leq RSoC(\Pi)$.
\end{theorem}
The proof of theorem~\ref{theorem:ppfpp_improves_rsoc} is trivial as the $ES$ will always contain the previous plan and SIPP is an optimal algorithm.

\section{Experimental Results}

\label{sec:experimental-results}

\plan{Introduce the section with what is inside the section}
In this section, we evaluate our algorithms on a standard \mapf benchmark \cite{stern2019multi} using the modified PIBT~\cite{okumura2022priority} and LaCAM$^*$~\cite{okumura2024lacam3} algorithms as described in Section~\ref{sec:subsolvers}. 

\commentout{
\plan{Define the FoV function we use in the experiments}
In this paper we use a $FoV$ function defined in the following way:
$F_i^r(v)=\{ u | |u.x - v.x| \leq r \And |u.y - v.y| \leq r \}$

The $FoV$ function we use has a parameter 
of $r \geq 0$ which is the $FoV$ radius.
Notice that $\forall_{v \in V}, \forall_{a_i, a_j}, \forall_{r \geq 0}: F_i^r(v)=F_j^r(v) \And F_i^r$ is symmetric.
Also notice that when $r=0 \Rightarrow F_i^0(v)=\{ v \}$, and as seen in corollary~\ref{cor:kppmapf-sub-ekppmapf}, in this case the \fpp solves a \kppmapf problem.
\begin{corollary}
    \label{cor:fov_radius_bigger_than_0_enables_ppfpp}
    When $r>0 \Rightarrow \forall v,u \text{ s.t. } u \in v.neighbors: u \in F_i^r(v) \Rightarrow$ we can use \ppfpp to improve \fpp generated plans' $RSoC$ according to theorem~\ref{theorem:ppfpp-solves-ekppmapf}.
\end{corollary}
}
\plan{Define the FoV function we use in the experiments}
We ran our experiments 4-neighborhood grids from the standard grid-based \mapf benchmark \cite{stern2019multi}. Twelve grids were selected from this benchmark with different size and sparseness.
The maps chosen to run experiments on, with their corresponding amount of vertices, are:
(1) brc202d ($|V|=43{,}151$), (2) lt\_gallowstemplar\_n ($|V|=10{,}021$), (3) maze-32-32-2 ($|V|=666$), (4) orz900d ($|V|=96{,}603$), (5) ost003d ($|V|=13{,}214$), (6) random-32-32-20 ($|V|=819$), (7) random-64-64-20 ($|V|=3{,}270$), (8) room-32-32-4 ($|V|=682$), (9) room-64-64-16 ($|V|=3{,}646$), (10) room-64-64-8 ($|V|=3{,}232$), (11) warehouse-20-40-10-2-1 ($|V|=22{,}599$), and (12) warehouse-20-40-10-2-2 ($|V|=38{,}756$).
In terms of sensing, we assumed an agent can sense nearby agent based on the following parametric and symmetric FoV function:
\begin{equation}
 F_i^r(v)=\{ u | |u.x - v.x| \leq r \And |u.y - v.y| \leq r \}   
\end{equation}

\subsection{Results}


\plan{experiments for \fpp over the entire maps and sub solvers}
In the first set of experiment, we varied the number of agents $N= 10, 20, \ldots, 50$, desired $k$-privacy values of $k=1,2,3$, and FoV sizes $r=1,2,$ and $3$. 
Each setting was run 3 times with different random seeds, and was given 1 minute to finish.
The experiments ran on a cluster of computers, each with 20 CPUs and 40GB of RAM, and in parallel with 4 processes running different instances.

Figure~\ref{fig:fpp_k} shows a cactus chart of the RSoC of the problem against the \#solved problems for using both sub solvers LaCAM* and PIBT in \fpp with different k values. It was analyzed over the entire experiment results, including all different agent amounts and $FoV \ r$ values.
The x axis represents the amount of solved instances, and the y axis represents the RSoC of the solved instance.
As expected, increasing $k$ decreases the amount of solved instances, for example, in figure~\ref{fig:fpp_k__random-64-64-20}, we saw that using LaCAM$^*$ with $k=1$ solved 56 instances, $k=2$ solved 37 and $k=3$ solved 29.
In most cases, using LaCAM$^*$ as the sub-solver of \fpp solved more instances, for example in figure~\ref{fig:fpp_k__brc202d} with $k=2$, we saw that LaCAM$^*$ solved 38 while PIBT solved only 28 instances.
Notice that the RSoC of the instances that were solved by LaCAM$^*$ and not by PIBT (the harder instances) was extremely high. This is due to the fact that LaCAM$^*$ is an any-time algorithm that improves the SoC of the found plan when more time is given, so on harder problems it found an initial plan and improved it only a little in comparison to easier problems. This trend can be seen for example in figures~\ref{fig:fpp_k__brc202d} and \ref{fig:fpp_k__random-64-64-20}.

Figure~\ref{fig:fpp_fov} shows a cactus chart similar to figure~\ref{fig:fpp_k}, but using different FoV radius values.
Increasing the FoV radius drastically impacts both the amount of solved instances and the solution quality. For example in the map maze-32-32-2 in figure~\ref{fig:fpp_fov__maze-32-32-2} the only solved instances were of FoV radius $=0$.
We also see that in figure~\ref{fig:fpp_fov__warehouse-20-40-10-2-2} it was not so drastically affected from the increase in FoV. We assume it is because of the many narrow corridors within the map that make the instances with FoV radius $> 0$ have a very little amount of $FoV$ conflicts, since they can just follow each other in the corridor in a safe distance from each other and they will be able to pass without a conflict.

Similar trends were observed in the other maps. We show these results for both figures~\ref{fig:fpp_k} and~\ref{fig:fpp_fov} in the supplementary material.




\begin{figure*}[t]
\centering
\begin{subfigure}[t]{0.23\textwidth}
\centering
  \includegraphics[width=1\linewidth]{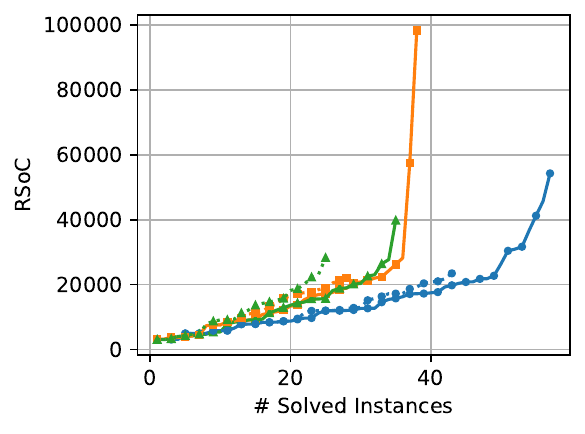}
  \caption{$brc202d$}
  \label{fig:fpp_k__brc202d}
  \Description{Cactus chart of RSoC for each configuration of k in each sub-solver (LaCAM/PIBT) over \# of solved instances for the map brc202d.}
\end{subfigure}\hspace{1em}
\begin{subfigure}[t]{0.23\textwidth}
\centering
  \includegraphics[width=1\linewidth]{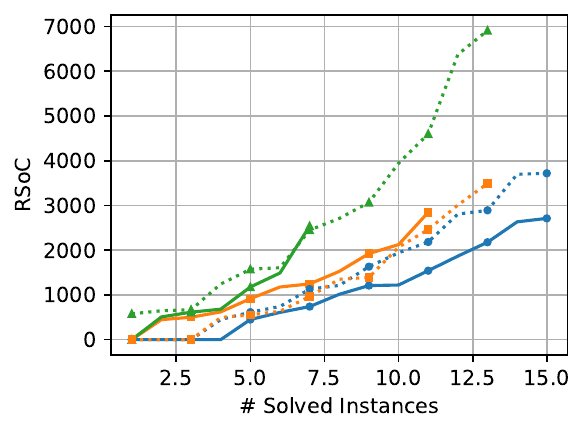}
  \caption{$maze-32-32-2$}
  \label{fig:fpp_k__maze-32-32-2}
  \Description{Cactus chart of RSoC for each configuration of k in each sub-solver (LaCAM/PIBT) over \# of solved instances for the map maze-32-32-2.}
\end{subfigure}\hspace{1em}
\begin{subfigure}[t]{0.23\textwidth}
\centering
  \includegraphics[width=1\linewidth]{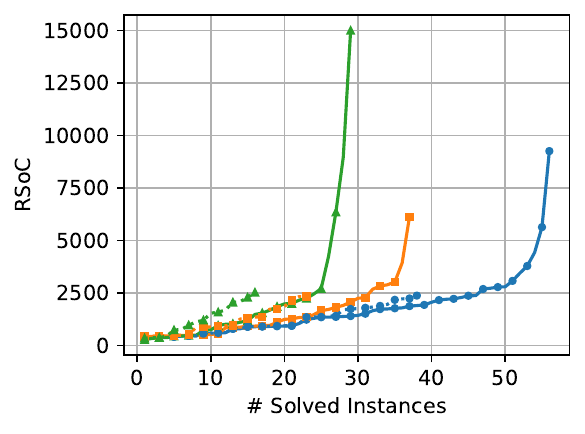}
  \caption{$random-64-64-20$}
  \label{fig:fpp_k__random-64-64-20}
  \Description{Cactus chart of RSoC for each configuration of k in each sub-solver (LaCAM/PIBT) over \# of solved instances for the map random-64-64-20.}
\end{subfigure}\hspace{1em}
\begin{subfigure}[t]{0.23\textwidth}
\centering
  \includegraphics[width=1\linewidth]{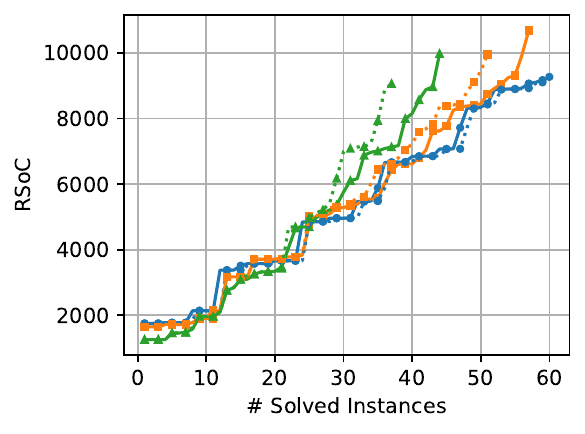}
  \caption{$warehouse-20-40-10-2-2$}
  \label{fig:fpp_k__warehouse-20-40-10-2-2}
  \Description{Cactus chart of RSoC for each configuration of k in each sub-solver (LaCAM/PIBT) over \# of solved instances for the map warehouse-20-40-10-2-2.}
\end{subfigure}\hspace{1em}
\begin{minipage}{\textwidth}
    \centering
    \includegraphics[width=0.7\linewidth]{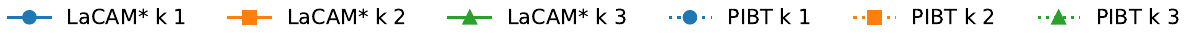}
    \label{fig:fpp_k__legend}
    \Description{Legend}
\end{minipage}
\caption{Cactus chart of $RSoC$ for each configuration of $k$ in each sub-solver $(LaCAM/PIBT)$ over \# of solved instances.}
\label{fig:fpp_k}
\end{figure*}

\begin{figure*}[t]
\centering
\begin{subfigure}[t]{0.23\textwidth}
\centering
  \includegraphics[width=1\linewidth]{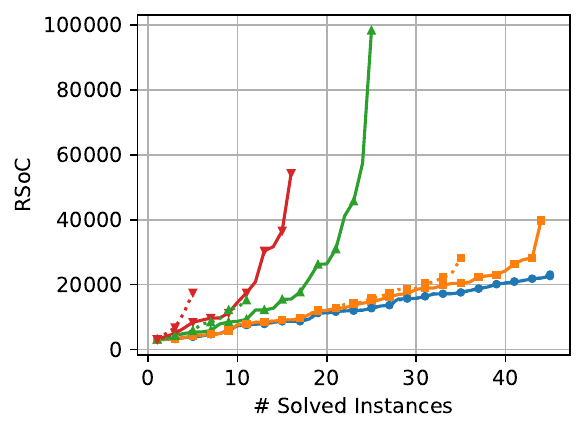}
  \caption{$brc202d$}
  \label{fig:fpp_fov__brc202d}
  \Description{Cactus chart of RSoC for each configuration of FoV in each sub-solver (LaCAM/PIBT) over \# of solved instances for the map brc202d.}
\end{subfigure}\hspace{1em}
\begin{subfigure}[t]{0.23\textwidth}
\centering
  \includegraphics[width=1\linewidth]{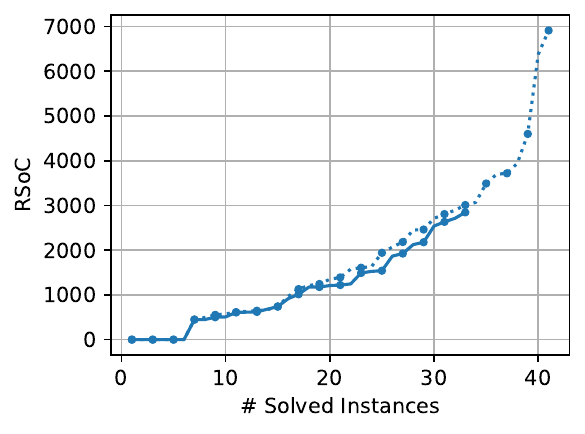}
  \caption{$maze-32-32-2$}
  \label{fig:fpp_fov__maze-32-32-2}
  \Description{Cactus chart of RSoC for each configuration of FoV in each sub-solver (LaCAM/PIBT) over \# of solved instances for the map maze-32-32-2.}
\end{subfigure}\hspace{1em}
\begin{subfigure}[t]{0.23\textwidth}
\centering
  \includegraphics[width=1\linewidth]{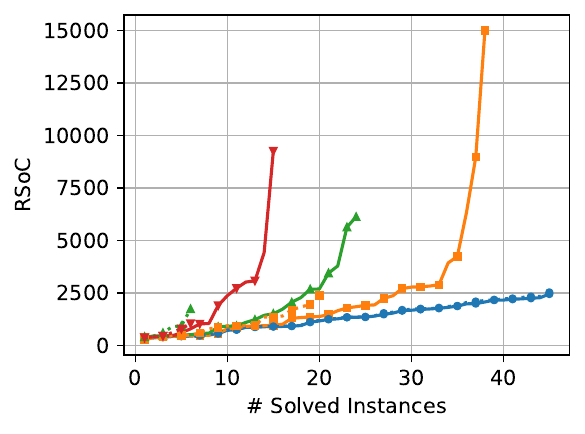}
  \caption{$random-64-64-20$}
  \label{fig:fpp_fov__random-64-64-20}
  \Description{Cactus chart of RSoC for each configuration of FoV in each sub-solver (LaCAM/PIBT) over \# of solved instances for the map random-64-64-20.}
\end{subfigure}\hspace{1em}
\begin{subfigure}[t]{0.23\textwidth}
\centering
  \includegraphics[width=1\linewidth]{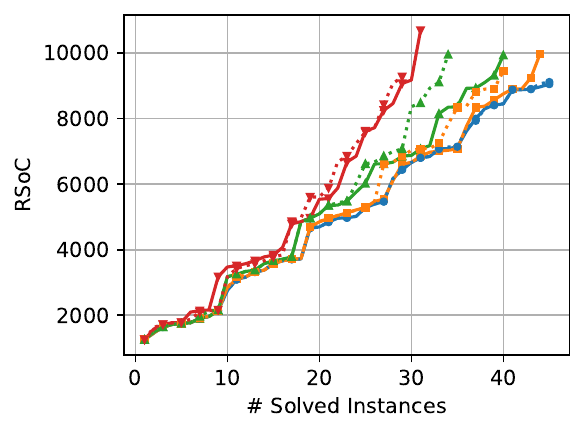}
  \caption{$warehouse-20-40-10-2-2$}
  \label{fig:fpp_fov__warehouse-20-40-10-2-2}
  \Description{Cactus chart of RSoC for each configuration of FoV in each sub-solver (LaCAM/PIBT) over \# of solved instances for the map warehouse-20-40-10-2-2.}
\end{subfigure}\hspace{1em}
\begin{minipage}{\textwidth}
    \centering
    \includegraphics[width=0.9\linewidth]{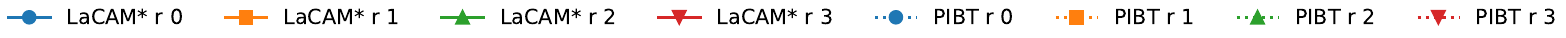}
    \label{fig:fpp_fov__legend}
    \Description{Legend}
\end{minipage}
\caption{Cactus chart of $RSoC$ for each configuration of $FoV$ in each sub-solver $(LaCAM/PIBT)$ over \# of solved instances.}
\label{fig:fpp_fov}
\end{figure*}

\begin{table}[ht]
\caption[\ppfpp improves \fpp RSoC]{Average \emph{RSoC improvement\%} $\left(\frac{RSoC(fPP)-RSoC(PPfPP)}{RSoC(fPP)}\right)*100\%$ (mean $\pm$ std, max and median (med)) of \ppfpp over \fpp, across domains with different k values. \#S is the amount of solved instances by \fpp (all instances also were also solved by \ppfpp). * and ** marks significant improvement of \ppfpp over \fpp (Wilcoxon, p $<$ 0.05 and p $<$ 0.01 accordingly). \fpp MS is the makespan of the original plan created by \fpp. Each map shows the amount of vertices in it $(|V|)$.}
\label{tab:ppfpp-rsoc}
\begin{tabular}{llllllll}
\toprule
Map & k & \multicolumn{3}{c}{$RSoC \ Improvement\%$} & \#S & time & \fpp \\
 &  & mean $\pm$ std & max & med &  & [s] & MS \\
\midrule
\multirow[c]{9}{*}{\rotatebox{90}{\shortstack{brc202d\\$|V|=43{,}151$}}} & $2$ & $0.77 \pm 2.47^{**}$ & $11.50$ & $0.02$ & $56$ & $113.6$ & $943.9$ \\
 & $3$ & $0.64 \pm 1.65^{**}$ & $7.25$ & $0.02$ & $55$ & $116.0$ & $975.4$ \\
 & $4$ & $1.04 \pm 2.69^{**}$ & $11.30$ & $0.02$ & $47$ & $121.2$ & $1022.4$ \\
 & $5$ & $1.06 \pm 2.66^{**}$ & $11.27$ & $0.00$ & $40$ & $125.3$ & $1041.3$ \\
 & $6$ & $1.20 \pm 2.68^{**}$ & $10.72$ & $0.02$ & $39$ & $124.0$ & $1040.9$ \\
 & $7$ & $1.80 \pm 3.90^{**}$ & $12.90$ & $0.03$ & $42$ & $131.0$ & $1094.5$ \\
 & $8$ & $2.04 \pm 4.45^{**}$ & $16.91$ & $0.07$ & $30$ & $137.6$ & $1170.1$ \\
 & $9$ & $3.21 \pm 6.40^{**}$ & $22.83$ & $0.11$ & $25$ & $143.4$ & $1203.6$ \\
 & $10$ & $1.01 \pm 2.69^{**}$ & $13.21$ & $0.07$ & $25$ & $143.8$ & $1312.6$ \\
\cline{1-8}
\multirow[c]{9}{*}{\rotatebox{90}{\shortstack{$\text{random-64-64-20}$\\$|V|=3{,}270$}}} & $2$ & $0.20 \pm 1.05^{**}$ & $7.75$ & $0.00$ & $56$ & $0.3$ & $90.5$ \\
 & $3$ & $0.62 \pm 1.89^{**}$ & $10.34$ & $0.00$ & $48$ & $0.3$ & $96.7$ \\
 & $4$ & $0.73 \pm 1.84^{**}$ & $7.65$ & $0.00$ & $38$ & $0.3$ & $100.3$ \\
 & $5$ & $1.33 \pm 2.60^{**}$ & $10.00$ & $0.23$ & $29$ & $0.4$ & $106.8$ \\
 & $6$ & $1.04 \pm 2.30^{**}$ & $8.33$ & $0.20$ & $21$ & $0.4$ & $120.4$ \\
 & $7$ & $0.70 \pm 1.24^{**}$ & $4.85$ & $0.19$ & $19$ & $0.4$ & $146.4$ \\
 & $8$ & $0.93 \pm 1.88^{**}$ & $5.72$ & $0.23$ & $15$ & $0.5$ & $151.8$ \\
 & $9$ & $2.97 \pm 5.36^{**}$ & $18.32$ & $0.16$ & $12$ & $0.7$ & $244.6$ \\
 & $10$ & $0.86 \pm 1.55^*$ & $3.64$ & $0.16$ & $5$ & $0.9$ & $294.6$ \\
\cline{1-8}
\multirow[c]{9}{*}{\rotatebox{90}{$\text{Mean}$}} & $2$ & $0.50 \pm 1.61^{**}$ & $11.50$ & $0.00$ & $518$ & $19.5$ & $293.3$ \\
 & $3$ & $0.68 \pm 2.22^{**}$ & $25.74$ & $0.00$ & $440$ & $22.3$ & $336.9$ \\
 & $4$ & $0.49 \pm 1.55^{**}$ & $11.30$ & $0.00$ & $360$ & $24.6$ & $374.7$ \\
 & $5$ & $0.93 \pm 2.35^{**}$ & $19.11$ & $0.02$ & $289$ & $26.4$ & $408.3$ \\
 & $6$ & $0.83 \pm 2.25^{**}$ & $17.06$ & $0.00$ & $244$ & $30.2$ & $451.8$ \\
 & $7$ & $1.04 \pm 2.86^{**}$ & $18.33$ & $0.03$ & $202$ & $37.7$ & $507.6$ \\
 & $8$ & $1.06 \pm 2.78^{**}$ & $16.91$ & $0.04$ & $176$ & $36.2$ & $518.6$ \\
 & $9$ & $1.66 \pm 3.90^{**}$ & $22.83$ & $0.05$ & $137$ & $38.6$ & $548.7$ \\
 & $10$ & $0.71 \pm 1.93^{**}$ & $13.21$ & $0.05$ & $119$ & $41.5$ & $630.2$ \\
\cline{1-8}
\bottomrule
\end{tabular}
\end{table}

In the second set of experiments, we aimed to explore the impact of higher values of $k$. Thus, we limited the range of values for the other experiment parameters, namely we ran \fpp on all maps with $k \in [2, 10]$, $N=10$, $r=1$, and using PIBT and LaCAM$^*$ as sub-solvers. 
Every experiment was performed 30 times for each configuration of parameters with different random seeds. 
For each instance solved by \fpp under 1 minute, we ran the \ppfpp algorithm with a timeout of 5 minutes. We ran this entire experiment on the cluster, with each map running in a different node. Each node ran with 40 CPUs for \fpp and 50 CPUs for \ppfpp and 60GB of RAM, in parallel, with 8 processes running different instances and 10 threads running in each process for calculating the \emph{ES} in parallel over different time steps.
The maps \emph{orz900d} and \emph{maze-32-32-2} are not shown in the results because of different reasons.
\emph{orz900d} does not appear since it is a very large map ($|V|=96,603$) and when running it in \ppfpp in our implementation it generates an \emph{out of memory} exception on the node for many problems.
\emph{maze-32-32-2} does not appear since no instance of it was solved by \fpp when using $r>0$.

Table~\ref{tab:ppfpp-rsoc} shows the results of the experiment for two interesting maps \emph{brc202d}, \emph{random-64-64-20} and the mean over all maps. The table contains results for both \mapf algorithms PIBT and LaCAM$^*$ together since there were similar trends for each of them separately. A table like this is shown in the supplementary material for the rest of the maps.
In almost all $k$ value for each map, the RSoC improvement is statistically significant.
In some cases, such as in $k=9$ in both maps, the maximal RSoC improvement was very high ($22.83\%$ for \emph{brc202d} and $18.32\%$ for \emph{random-64-64-20}).
As the $k$ value grows, the mean improvement grows as well, as can be seen in both maps, and in the mean of all maps, notice that the median value also grows with k, although it is very small. This consists with our hypothesis where larger $k$ values will result in larger RSoC improvement.
The runtime of \ppfpp depends directly on the map size, where larger maps take longer to generate the \emph{ES}. This can be seen for example in map \emph{brc202d} where $|V|=43,151$, the runtime is extremely high, in comparison to map \emph{random-64-64-20} where $|V|=3,270$, and the runtime there is below one second for all k values.
The runtime also depends on the makespan of the original plan generated by \fpp, where the higher the makespan, the higher the runtime of \ppfpp. This is because we need to calculate ES for each time step in the original plan, and go over all of the vertices. It can be seen in both maps and in the mean over all maps.

\section{Related Work}

\plan{Related work on privacy preserving planning}

In the field of Collaborative Privacy Preserving Planning (\cppp) there has been a lot of work regarding privacy in multi-agent planning scenarios \cite{maliah2017collaborative,maliah2016stronger,lehman2021partial,lev2022reducing,nissim2014distributed,brafman2015privacy}. It is a similar field, but privacy preserving \mapf is a special case of the \cppp general field. It is important to study privacy in \mapf as well, since it allows us to handle the \mapf case with better suited algorithms than the general purpose algorithms used in \cppp. In the \cppp field the privacy is measured by not disclosing the variables and actions of the world of each agent to the other agents (or disclose part of them keeping the other part private) - in privacy preserving \mapf --- which we properly define in this paper for the first time, the actions are known to all agents (move to any neighbor or stay in place). However, the vertices that each agent is in at any given time is a private knowledge that needs to be kept from other agents while not conflicting with other agents in the same time - which is a hard task to achieve. Brafman~\cite{brafman2015privacy} suggests a strong type of privacy where a value of a variable is strongly private if the other agents cannot deduce its existence from the information available to them.
In this paper we suggest a privacy preserving \mapf problem where we do not seek a form of fully strong privacy on the path of each agent, but we seek that the other agents cannot deduce the location of the agent in a specific time step exactly but from within a minimum of k possible values.

\plan{Related work on privacy in DCOP}
\balance
Another well-studied privacy preserving field is the privacy preserving distributed constraint optimization problem (DCOP) \cite{kogan2022privacy,grinshpoun2016p,greenstadt2007ssdpop,faltings2008privacy}, where the goal is to solve distributed combinatorial problems in which the variables of the problem are owned by different agents while keeping privacy over the constraints and variables of the different agents. 
In this paper, we focus on the privacy preserving \mapf problem we suggested, where, unlike privacy preserving DCOP, the goal is to plan paths for a set of agents, instead of optimizing some utility function.

\plan{Related work on the paper "Finding all optimal solutions in MAPF" (Dor Atzmon)}
Privacy has rarely been mentioned in the context of \mapf. 
One exception is recent work on finding all optimal solutions in \mapf \cite{bardugo2025finding}. They state that finding all solutions can help keep privacy about the chosen paths. However, the privacy considered there is with respect to an external agent, and not of hiding the paths of agents from each other as we do in this work.

\section{Conclusion and Future Work}

\plan{summarize the novel problems and approaches stated in this paper}

This paper introduces a new field in \mapf where privacy is considered.
We defined two novel privacy preserving problems, the \kppmapf which is a problem where agents must find paths that do not collide, without sharing their real path during planning, and the \ekppmapf which is an extension of the \kppmapf with the addition of also avoiding the paths of agents to be disclosed during the execution stage when there are limited sensors on the agents by keeping out of the field of view of other agents during planning.
We also provide two novel algorithms to solve the new problems, \kpp, which solves the \kppmapf problem by introducing \emph{mock agents} to confuse other agents during planning and the \fpp, which extends the \kpp in order to solve the \ekppmapf problem by also adapting the sub-solver used to keep out of the field of view of agents in different groups.

\plan{say about the \ppfpp algorithm that improves RSoC}

We also introduce the \ppfpp algorithm, which improves the cost of plans generated for the \ekppmapf problem by defining a concept of \emph{safe-zones} where each agent group can walk freely in and re-plan for the single real agent inside of it.

\plan{extensive theory and experiments done}

We also provide extensive theoretical and experimental evaluation of our algorithms, showing they preserve k-privacy and Runtime k-privacy and that the \ppfpp algorithm improves the cost of the \ekppmapf plans resulted by \fpp.

\plan{Future work in \ppfpp}

In \ppfpp, one can further research more complicated functions other than the random function suggested in algorithm~\ref{alg:extend-safe-zone-random}, that follows the rules from definition~\ref{def:rules-for-extended-safe-zone}, for improving the extension of safe zones, and maybe improve the $RSoC(\Pi)$ even further.

Another future research direction is to explore other sub-planners to use in the \fpp algorithm, such as optimal \mapf planners, and see whether the \ppfpp improves their cost significantly, and makes it closer to the optimal RSoC.

\plan{Add swap to pibt to improve performance}
In our work we used a vanilla version of PIBT which does not use the \emph{swap} improvement introduced at \cite{okumura2023lacam2}. A future research direction could be to support \emph{FoV} in the \emph{swap} operation of PIBT to improve the planning time of \fpp while using LaCAM$^*$, and solving more instances while using PIBT.

\bibliographystyle{ACM-Reference-Format} 
\bibliography{main}


\end{document}